\documentclass[nohyperref]{article}

\usepackage{microtype}
\usepackage{graphicx}
\usepackage{booktabs} %

\usepackage{hyperref}

\usepackage[accepted]{icml2022}

\usepackage{amsmath}
\usepackage{amssymb}
\usepackage{mathtools}
\usepackage{amsthm}

\usepackage[capitalize,noabbrev]{cleveref}

\usepackage[textsize=tiny]{todonotes}

\usepackage{multirow,array} %
\usepackage[position=bottom]{subcaption} %
\usepackage{wrapfig} %

\newcommand{\tabref}[1]{Table~\ref{#1}}

\newcommand{\secref}[1]{Section~\ref{#1}}
\newcommand{\appref}[1]{Appendix~\ref{#1}}

\newcommand{\Gg}{\mathcal{G}}

\DeclareMathOperator*{\argmax}{argmax}

\theoremstyle{definition}
\usepackage{thmtools, thm-restate}
\usepackage{nameref,
}
\declaretheorem[numberwithin=subsection,
                name=Theorem,
                refname={theorem,theorems},
                Refname={Theorem,Theorems}]{theorem}

\declaretheorem[sibling=theorem,
                name=Example,
                refname={example,examples},
                Refname={Example,Examples}]{example}

\icmltitlerunning{For Learning in Symmetric Teams, Local Optima are Global Nash Equilibria}

\begin{document}

\twocolumn[
\icmltitle{For Learning in Symmetric Teams, \\
           Local Optima are Global Nash Equilibria}

\icmlsetsymbol{equal}{*}

\begin{icmlauthorlist}
\icmlauthor{Scott Emmons}{ucb}
\icmlauthor{Caspar Oesterheld}{cmu}
\icmlauthor{Andrew Critch}{ucb}
\icmlauthor{Vincent Conitzer}{duke}
\icmlauthor{Stuart Russell}{ucb}
\end{icmlauthorlist}

\icmlaffiliation{ucb}{University of California, Berkeley}
\icmlaffiliation{duke}{Duke University, University of Oxford}
\icmlaffiliation{cmu}{Carnegie Mellon University}

\icmlcorrespondingauthor{Scott Emmons}{emmons@berkeley.edu}

\icmlkeywords{common-payoff game, identical interest game, team theory, Nash equilibrium, symmetric strategy profile}

\vskip 0.3in
]

\printAffiliationsAndNotice{}  %

\begin{abstract}
Although it has been known since the 1970s that a \textit{globally} optimal strategy profile in a common-payoff game is a Nash equilibrium, global optimality is a strict requirement that limits the result's applicability. In this work, we show that any \textit{locally} optimal symmetric strategy profile is also a (global) Nash equilibrium. Furthermore, we show that this result is robust to perturbations to the common payoff and to the local optimum. Applied to machine learning, our result provides a global guarantee for any gradient method that finds a local optimum in symmetric strategy space. While this result indicates stability to \textit{unilateral} deviation, we nevertheless identify broad classes of games where mixed local optima are unstable under \textit{joint}, asymmetric deviations. We analyze the prevalence of instability by running learning algorithms in a suite of symmetric games, and we conclude by discussing the applicability of our results to multi-agent RL, cooperative inverse RL, and decentralized POMDPs.
\end{abstract}

\section{Introduction}
\label{submission}

We consider {\em common-payoff} games (also known as {\em identical interest} games \citep{ui2009bayesian}), in which the payoff to all players is always the same. %
Such games model a wide range of situations involving cooperative action towards a common goal. Under the heading of {\em team theory}, they form an important branch of economics \citep{marschak1955elements,marschak1972economic}. In cooperative AI \citep{dafoe2021cooperative}, the common-payoff assumption holds in \textit{Dec-POMDPs} \citep{oliehoek2016concise}, where multiple agents operate independently according to policies designed centrally to achieve a common objective. Many applications of {\em multiagent reinforcement learning} also assume a common payoff \citep{foerster2016learning,foerster2018counterfactual,gupta2017cooperative}. Finally, in {\em assistance games} \citep{russell2019human} (also known as cooperative inverse reinforcement learning or CIRL games \citep{Hadfield-Menell+al:2017a}), which include at least one human and one or more ``robots,'' it is assumed that the robots' payoffs are exactly the human's payoff, even if the robots must learn it.

Our focus is on symmetric strategy profiles in common-payoff games. Loosely speaking, a symmetric strategy profile is one in which some subset of players share the same strategy; \secref{preliminaries-section} defines this in a precise sense. For example, in Dec-POMDPs, an offline solution search may consider only symmetric strategies as a way of reducing the search space. (Notice that this does not lead to identical {\em behavior}, because strategies are state-dependent.) In common-payoff multiagent reinforcement learning, each agent may collect percepts and rewards independently, but the reinforcement learning updates can be pooled to learn a single parameterized policy that all agents share; prior work has found experimentally that ``parameter sharing is crucial for reaching the optimal protocol'' \citep{foerster2016learning}. In team theory, it is common to develop a strategy that can be implemented by every employee in a given category and leads to high payoff for the company. In civic contexts, symmetry commonly arises through notions of fairness and justice. In treaty negotiations and legislation that mandates how parties behave, for example, there is often a constraint that all parties be treated equally.

For the purposes of this paper, we consider Nash equilibria---strategy profiles for all players from which no individual player has an incentive to deviate---as a reasonable solution concept. \citet{marschak1972economic} make the obvious point that a globally optimal (possibly asymmetric) strategy profile---one that achieves the highest common payoff---is necessarily a Nash equilibrium. Moreover, it can be found in time linear in the size of the payoff matrix.

\begin{table*}[t]
    \begin{subfigure}{0.32\textwidth}
    \centering
    \begin{tabular}{cc|c|c|}
      & \multicolumn{1}{c}{} & \multicolumn{2}{c}{Mobile}\\
      & \multicolumn{1}{c}{} & \multicolumn{1}{c}{H}  & \multicolumn{1}{c}{W} \\\cline{3-4}
      \multirow{2}*{Auto}  & H & 1 & 0 \\\cline{3-4}
                             & W & 2 & 1 \\\cline{3-4}
    \end{tabular}
    \caption{Taxis have different permits}
        \label{laundry-washing-asymmetric}
    \end{subfigure}
    \begin{subfigure}{0.32\textwidth}
    \centering
    \begin{tabular}{cc|c|c|}
      & \multicolumn{1}{c}{} & \multicolumn{2}{c}{Mobile}\\
      & \multicolumn{1}{c}{} & \multicolumn{1}{c}{H}  & \multicolumn{1}{c}{W} \\\cline{3-4}
      \multirow{2}*{Auto}  & H & 1 & 2 \\\cline{3-4}
                             & W & 2 & 1 \\\cline{3-4}
    \end{tabular}
    \caption{Taxis are identical}
        \label{laundry-washing-symmetric-unstable}
    \end{subfigure}
    \begin{subfigure}{0.32\textwidth}
    \centering
    \begin{tabular}{cc|c|c|}
      & \multicolumn{1}{c}{} & \multicolumn{2}{c}{Mobile}\\
      & \multicolumn{1}{c}{} & \multicolumn{1}{c}{H}  & \multicolumn{1}{c}{W} \\\cline{3-4}
      \multirow{2}*{Auto}  & H & 1 & 0 \\\cline{3-4}
                             & W & 0 & 1 \\\cline{3-4}
    \end{tabular}
    \caption{Large groups need both taxis}
        \label{laundry-washing-symmetric-stable}
    \end{subfigure}
\caption{Three versions of the self-driving taxi game. Solutions are described in the text.}
\label{tab:laundry_washing_up_game}
\end{table*}

In any sufficiently complex game, however, we should not expect to be able to find a \textit{globally} optimal strategy profile. For example, matrix games have size exponential in the number of players, and the matrix representation of a game tree has size exponential in the depth of the tree. Therefore, global search over all possible contingency plans is infeasible for all but the smallest of games. This is why some of the most effective methods in machine learning, such as gradient methods, employ \textit{local} search over strategy space.

Lacking global guarantees, local search methods may converge only to \textit{locally} optimal strategy profiles. Roughly speaking, a \textit{locally optimal} strategy profile is a strategy profile from which no group of players has an incentive to slightly deviate. Obviously, a locally optimal profile may not be a Nash equilibrium, as a player may still have an incentive to deviate to some more distant point in strategy space. Nonetheless, \citet{ratliff:2016aa} argue that a local Nash equilibrium may still be stable in a practical sense if agents are computationally unable to find a better strategy.

The central question of this work is: what can we say about the (global) properties of locally optimal symmetric strategy profiles? Our first main result, informally stated, is that {\em in a symmetric, common-payoff game, every local optimum in symmetric strategies is a (global) Nash equilibrium}. \secref{sec:theoretical} states the result more precisely and gives an example illustrating its generality. \secref{sec:robustness} shows that the result is robust to perturbations to the common payoff and to the local optimum. \secref{sec:symmetry_illustration} elaborates on the symmetry required by the result, illustrating how the theorem applies even when the physical environment is asymmetric and when players have differing capabilities. Complete proofs for all of our results are in the appendices.

Despite decades of research on symmetry in common-payoff games \citep{sandholm2001potential,brandt2009symmetries}, our result appears to be novel. There are some echoes of the result in the literature on single-agent decision making \citep{piccione1997interpretation,Briggs2010,Schwarz2015}, which can be connected to symmetric solutions of common-payoff games by treating all players jointly as a single agent, but our result appears more general than published results. Perhaps closest to our work is \citet{piccione1997interpretation}, which establishes an equilibrium-of-sorts among the ``modified multi-selves'' of a single player's information set. The proof we give of our result has similarities with the proof (of a related but different result) in \citet{Taylor2016}.

In the second half of our paper, we turn to the thorny question of \textit{stability}. Instability, if not handled carefully, might lead to major coordination failures in practice \citep{bostrom2016unilateralist}. It is already known that local strict optima in a totally symmetric team game attain one type of stability, but the issue is complex because there are several ways of enforcing (or not enforcing) strict symmetries in payoffs and strategies \citep{milchtaich2016static}. Whereas our first main result implies stability to \textit{unilateral} deviation, our second main result establishes when stability exists to \textit{joint}, possibly-asymmetric,  deviation. We prove for a non-degenerate class of games that local optima in symmetric strategy space fail to be local optima in asymmetric strategy space if and only if at least one player is mixing, and we experimentally quantify how often mixing occurs for learning algorithms in the GAMUT suite of games \citep{nudelman2004run}.

\section{Motivating Examples}
\label{sec:laundry_washing_up}

\begin{figure}[t]
\centering
\includegraphics[width=1.0\columnwidth]{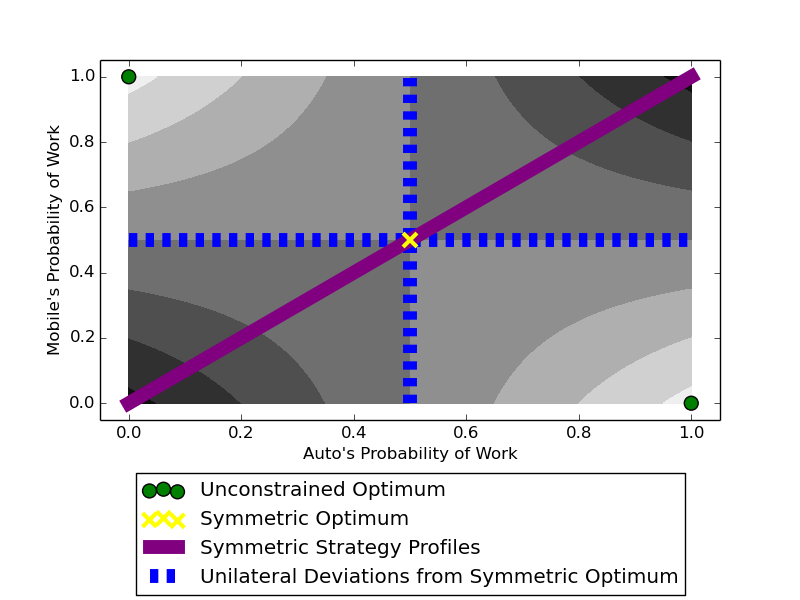}
\caption{The strategy profile landscape of the symmetric self-driving taxi game (Table~\protect\ref{laundry-washing-symmetric-unstable}). Lighter color is higher expected utility. Although the symmetric optimum has lower expected utility than the unrestricted optima, total symmetry of the game implies that the symmetric optimum is a Nash equilibrium; this is a special case of \Cref{thm:efg_equilibrium}.}
\label{fig:anticoordination_landscape}
\end{figure}

To gain some intuition for these concepts and claims, let us consider a situation with two self-driving taxis, Auto and Mobile. Two groups of people need rides: one group needs to go home ($H$), and the other needs to go to work ($W$). It is evident that a symmetric strategy profile---both taxis driving home or both driving to work---is not ideal, because the other trip will not get made.

The first version of the game, whose payoffs $U$ are shown in \tabref{laundry-washing-asymmetric}, is asymmetric: 
Auto only has a work entrance permit, whereas Mobile only has a home entrance permit.
Here, as \citet{marschak1972economic} pointed out, the strategy profile $(W,H)$
is both globally optimal and a Nash equilibrium. If we posit a mixed (randomized) strategy profile in which
Auto and Mobile have work probabilities $p$ and $q$ respectively, the gradients $\partial U/\partial p$ and $\partial U/\partial q$ are
$+1$ and $-1$, driving the solution to $(W,H)$.

In the second version of the game
(\tabref{laundry-washing-symmetric-unstable}), both taxis have both permits, and symmetry is restored. The pure profiles
$(H,W)$ and $(W,H)$ are (asymmetric) globally optimal solutions and
hence Nash equilibria. Figure~\ref{fig:anticoordination_landscape}
shows the entire payoff landscape as a function of $p$ and $q$:
looking just at symmetric strategy profiles, it turns out that there
is a local optimum at $p=q=0.5$, i.e., where Auto and Mobile toss fair
coins to decide what to do. Although the expected payoff of this solution is
lower than that of the asymmetric optima, the local optimum is, nonetheless,
a Nash equilibrium. All unilateral deviations from the symmetric local optimum result in the same expected payoff because if one taxi is tossing a coin, the other taxi can do nothing to improve the final outcome.

In the third version of the game
(\tabref{laundry-washing-symmetric-stable}), both the home and work groups of people are large and need both taxis. In this case, there is again a Nash equilibrium
at $p=q=0.5$, but it is a local minimum rather than a local maximum in symmetric strategy space.
Thus, not all symmetric Nash equilibria are symmetric local optima; this is because
Nash equilibria depend on {\em unilateral} deviations, whereas symmetric local optima
depend on {\em joint} deviations that maintain symmetry.

\subsection{Complex Coordination Example where a Simple Symmetric Strategy is Best}
\label{sec:complex_coordination_example}

Consider 10 robots that must each choose between 3 actions, $a$, $b$, and $c$. If all robots play action $a$, they receive a reward of 1. If exactly one  robot plays action $b$ while the rest play action $c$, they receive a reward of $1 + \epsilon$. Otherwise, the reward is $0$. For small enough $\epsilon$, the optimal symmetric policy is for all robots to play action $a$. Here, trying to coordinate in symmetric strategies to reach the asymmetric optimum is suboptimal---the best symmetric strategy is the simple one. Furthermore, our subsequent theory shows that the best symmetric strategy is stable; it is locally optimal even when considering \textit{joint} (possibly asymmetric) deviations.

\section{Preliminaries: Games and Symmetries} \label{preliminaries-section}

\subsection{Normal-form Games}

Throughout, we consider \textit{normal-form games} $\Gg = (N, A, u)$ defined by a finite set $N$ with $\lvert N \rvert = n$ players, a finite set of action profiles $A = A_1 \times A_2 \times \ldots \times A_n$ with $A_i$ specifying the actions available to player $i$, and the utility function $u = (u_1, u_2, \ldots, u_n)$ with $u_i : A \to \mathbb{R}$ giving the utility for each player $i$ \citep{shoham2008multiagent}. We call $\Gg$ \textit{common-payoff} if $u_i(a) = u_j(a)$ for all action profiles $a \in A$ and all players $i, j$. In common-payoff games we may omit the player subscript $i$ from utility functions.

We model each player as employing a (mixed) \textit{strategy} $s_i \in \Delta(A_i)$, a probability distribution over actions. We denote the support of the probability distribution $s_i$ by $\mathrm{supp}(s_i)$. Given a (mixed) \textit{strategy profile} $s = (s_1, s_2, \ldots, s_n)$ that specifies a strategy for each player, player $i$'s expected utility is $EU_{i}(s) = \sum_{a \in A} u_i(a) \prod_{j=1}^n s_j(a_j)$. If a strategy $s_i$ for player $i$ maximizes expected utility given the strategies $s_{-i}$ of all the other players, i.e., if $s_i \in \argmax_{s_i'\in \Delta(A_i)} EU_{i}(s_i', s_{-i})$, we call $s_i$ a \textit{best response} to $s_{-i}$. If each strategy $s_i$ in a strategy profile $s$ is a best response to $s_{-i}$, we call $s$ a \textit{Nash equilibrium}. A Nash equilibrium $s$ is \textit{strict} if every $s_i$ is the unique best response to $s_{-i}$.

Note that, while we have chosen to use the normal-form game representation for simplicity, normal-form games are highly expressive. Normal-form games can represent mixed strategies in all finite games, including games with sequential actions, stochastic transitions, and partial observation such as imperfect-information extensive form games with perfect recall, Markov games, and Dec-POMDPs. To represent a sequential game in normal form, one simply lets each normal-form action be a complete strategy (contingency plan) accounting for every potential game decision.

\subsection{Symmetry in Game Structure}

We adopt the fairly general group-theoretic notions of symmetry introduced by \citet{von1944theory} and \citet{nash1951non}, and we borrow notation from \citet{plan2017symmetric}. More recent work has analyzed {\em narrower} notions of symmetry
\citep{reny1999existence,vester2012symmetric,milchtaich2016static,li2020structure}. For example, \citet{daskalakis2007anonymous}
study ``anonymous games'' and show that
anonymity substantially reduces the complexity of finding
solutions. Additionally, \citet{ham2013notions} generalizes the
player-based notion of symmetry to include further symmetries revealed
by renamings of actions. We conjecture our results extend to this more
general case, at some cost in notational complexity, but we leave this
to future work.

Our basic building block is a \textit{symmetry of a game}:

\begin{restatable}{definition}{defsymmetryofgame}
\label{def:symmetry}
Call a permutation of player indices $\rho: \{ 1, 2, ..., n \} \to \{ 1, 2, ..., n \}$ a \textit{symmetry of a game} $\Gg$ if, for all strategy profiles $(s_1, s_2, ..., s_n)$, permuting the strategy profile permutes the expected payoffs: $EU_{\rho(i)}((s_1, s_2, ..., s_n)) = EU_i((s_{\rho(1)}, s_{\rho(2)}, ..., s_{\rho(n)})),\ \forall i.$

Note that, when we speak of a symmetry of a game, we implicitly assume $A_i = A_j$ for all $i, j$ with $\rho(i) = j$ so that permuting the strategy profile is well-defined.\footnote{\label{ftn:action-mappings} We make this choice to ease notational burden, but we conjecture that our results can be generalized to allow for mappings between actions \citep{ham2013notions}, which we leave for future work.}
\end{restatable}

We characterize the symmetric structure of a game by its set of game symmetries:

\begin{restatable}{definition}{defgamesymmetries}
Denote the set of all symmetries of a game $\Gg$ by:
$\Gamma(\Gg) = \{\text{$\rho: \{ 1, 2, ..., n \} \to \{ 1, 2, ..., n \}$ a symmetry of $\Gg$}\}$.
\end{restatable}

A spectrum of game symmetries is possible. On one end of the spectrum, the identity permutation might be the only symmetry for a given game. On the other end of the spectrum, all possible permutations might be symmetries for a given game. Following the terminology of \citet{von1944theory}, we call the former case \textit{totally unsymmetric} and the latter case \textit{totally symmetric}:

\begin{restatable}{definition}{defspecialsymmetrycases}
\label{def:total_symmetry}
If $\Gamma(\Gg) = S_n$, the full symmetric group, we call the game $\Gamma(\Gg)$ \textit{totally symmetric}. If $\Gamma(\Gg)$ contains only the identity permutation, we call the game \textit{totally unsymmetric}.
\end{restatable}

Let $\mathcal{P} \subseteq \Gamma(\Gg)$ be any subset of the game symmetries. Because $\Gamma(\Gg)$ is closed under composition, we can repeatedly apply permutations in $\mathcal{P}$ to yield a group of game symmetries $\langle \mathcal{P} \rangle$:

\begin{restatable}{definition}{defgeneratedgroup}
Let $\mathcal{P} \subseteq \Gamma(\Gg)$ be a subset of the game symmetries. The group \textit{generated} by $\mathcal{P}$, denoted $\langle \mathcal{P} \rangle$, is the set of all permutations that can result from (possibly repeated) composition of permutations in $\mathcal{P}$: $\langle \mathcal{P} \rangle = \{ \rho_1 \circ \rho_2 \circ \ldots \circ \rho_m\ |\ m \in \mathbb{N}, \rho_1, \rho_2, \ldots, \rho_m \in \mathcal{P} \}$.
\end{restatable}

Group theory tells us that $\langle \mathcal{P} \rangle$ defines a closed binary operation (permutation composition) including an identity and inverse maps, and $\langle \mathcal{P} \rangle$ is the closure of $\mathcal{P}$ under function composition.

With a subset of game symmetries $\mathcal{P} \subseteq \Gamma(\Gg)$ in hand, we can use the permutations in $\langle \mathcal{P} \rangle$ to carry one player index to another. For each player $i$, we give a name to the set of player indices to which permutations in $\langle \mathcal{P} \rangle$ can carry $i$: we call it player $i$'s \textit{orbit}.

\begin{restatable}{definition}{deforbit}
Let $\mathcal{P} \subseteq \Gamma(\Gg)$ be a subset of the game symmetries $\Gamma(\Gg)$.
The \textit{orbit} of player $i$ under $\mathcal{P}$ is the set of all other player indices that $\langle \mathcal{P} \rangle$ can assign to $i$: $\mathcal{P}(i) = \{ \rho(i)\ |\ \rho \in \langle \mathcal{P} \rangle \}$.
\end{restatable}

By standard group theory, the orbits of a group action on a set partition the set's elements, so:

\begin{restatable}{proposition}{prporbitequivs}
\label{prp:orbit_equivs}
Let $\mathcal{P} \subseteq \Gamma(\Gg)$. The orbits of $\mathcal{P}$ partition the game's players.
\end{restatable}

\Cref{prp:orbit_equivs} tells us each $\mathcal{P} \subseteq \Gamma(\Gg)$ yields an equivalence relation among the players. To gain intuition for this equivalence relation, consider two extreme cases. In a totally unsymmetric game, $\Gamma(\Gg)$ contains only the identity permutation, in which case each player is in its own orbit of $\Gamma(\Gg)$; the equivalence relation induced by the orbit partition shows that no players are equivalent. In a totally symmetric game, by contrast, every permutation is an element of $\Gamma(\Gg)$, i.e., $\Gamma(\Gg) = S_n$, the full symmetric group; now, all the players share the same orbit of $\Gamma(\Gg)$, and the equivalence relation induced by the orbit partition shows that all the players are equivalent.

We leverage the orbit structure of an arbitrary $\mathcal{P} \subseteq \Gamma(\Gg)$ to define an equivalence relation among players because it adapts to however much or little symmetry is present in the game. Between the extreme cases of no symmetry ($n$ orbits) and total symmetry (1 orbit) mentioned above, there could be any intermediate number of orbits of $\mathcal{P}$. Furthermore, it might not be the case that players who share an orbit can be swapped in arbitrary ways. For an example of this, see Appendix~\ref{app:general_symmetry_example}.

\subsection{Symmetry in Strategy Profiles}

Having formalized a symmetry of a game in the preceding section, we follow \citet{nash1951non} and define symmetry in strategy profiles with respect to symmetry in game structure:

\begin{restatable}{definition}{defpinvariant}
Let $\mathcal{P} \subseteq \Gamma(\Gg)$ be a subset of the game symmetries $\Gamma(\Gg)$.
We call a strategy profile $s = (s_1, s_2, ..., s_n)$ $\mathcal{P}$\textit{-invariant} if $(s_1, s_2, ..., s_n) = (s_{\rho(1)}, s_{\rho(2)}, ..., s_{\rho(n)})$ for all $\rho \in \langle \mathcal{P} \rangle$.
\end{restatable}

The equivalence relation among players induced by the orbit structure of $\mathcal{P}$ is fundamental to our definition of symmetry in strategy profiles by the following proposition:

\begin{restatable}{proposition}{prpsyminveq}
\label{prp:sym_inv_eq}
A strategy profile $s = (s_1, s_2, ..., s_n)$ is $\mathcal{P}$-invariant if and only if $s_i = s_j$ for each pair of players $i$ and $j$ with $\mathcal{P}(i) = \mathcal{P}(j)$.
\end{restatable}

To state \Cref{prp:sym_inv_eq} another way, a strategy profile is $\mathcal{P}$-invariant if all pairs of players $i$ and $j$ that are equivalent under the orbits of $\mathcal{P}$ play the same strategy.

\begin{figure*}[t]
    \centering
    \begin{subfigure}[t]{0.33\textwidth}
        \centering
        \includegraphics[width=1.0\linewidth]{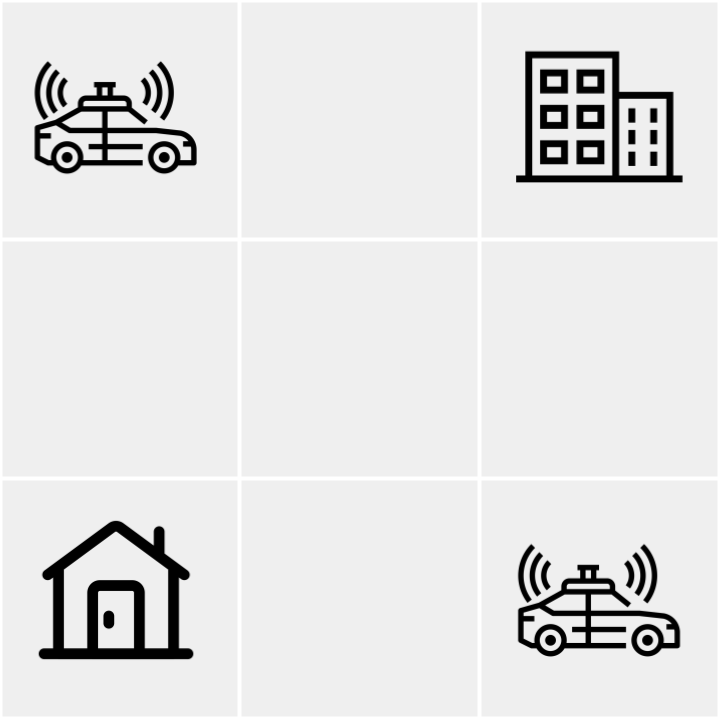}
        \caption{Symmetric agents \& environment}
        \label{fig:symmetric_mdp_a}
    \end{subfigure}\hfill
    \begin{subfigure}[t]{0.33\textwidth}
        \centering
        \includegraphics[width=1.0\linewidth]{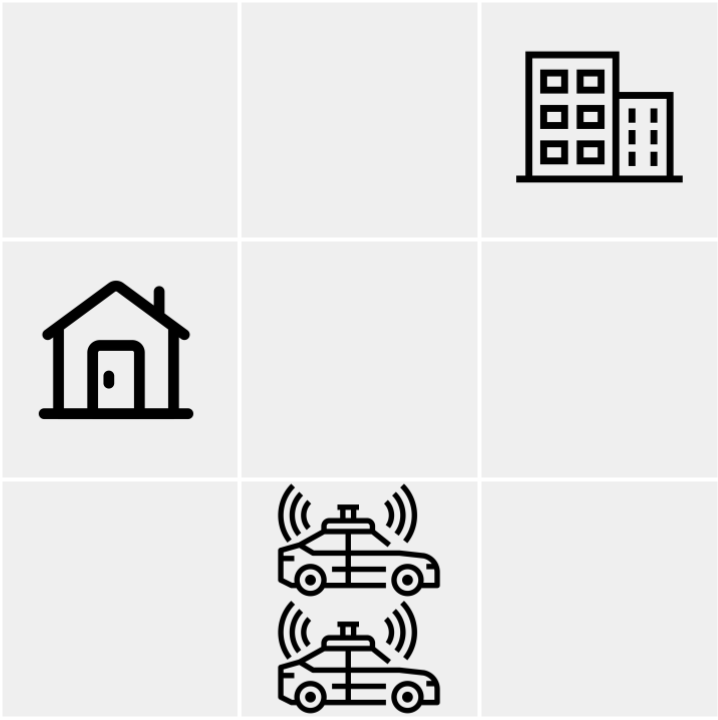}
        \caption{Same initial conditions}
        \label{fig:symmetric_mdp_b}
    \end{subfigure}\hfill
    \begin{subfigure}[t]{0.33\textwidth}
        \centering
        \includegraphics[width=1.0\linewidth]{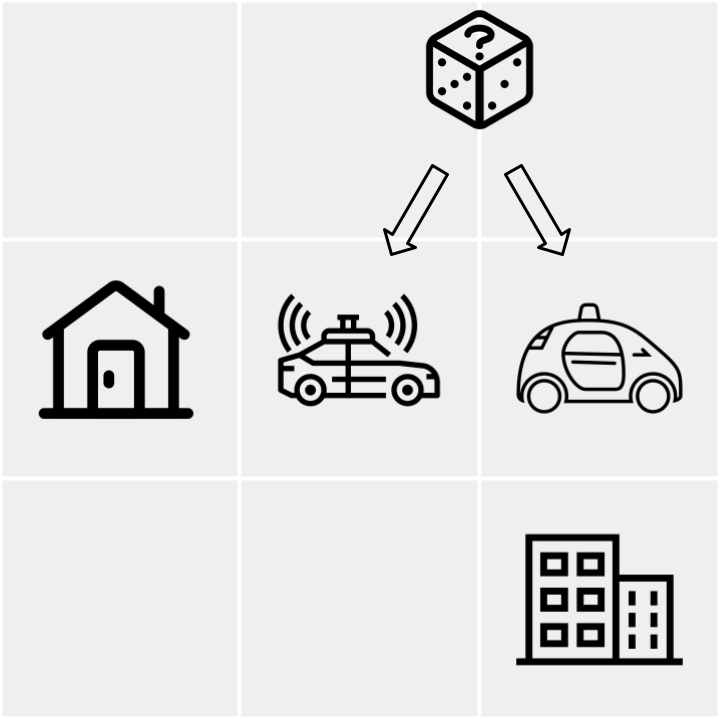}
        \caption{Veil of ignorance}
        \label{fig:symmetric_mdp_c}
    \end{subfigure}
\caption{Various self-driving taxi grid-world games that satisfy our symmetry requirement. (a) Symmetric agents in a symmetric environment. (b) Although the environment is asymmetric, the game is still symmetric because the agents have the same initial condition. (c) When agents must be programmed \textit{before} knowing their initial conditions (e.g. location, type), symmetry holds behind the \textit{veil of ignorance} (\secref{sec:veil_of_ignorance}) even with nonidentical agents and asymmetric environments.}
\label{fig:symmetry_illustration}
\end{figure*}

\subsection{Symmetry via the Veil of Ignorance} \label{sec:veil_of_ignorance}

Sometimes strategies must be specified for all players before knowing the players' roles and initial conditions. Consider writing laws or programming household robots; all players are treated equally in specifying situation-dependent contingency plans. When all players have equal likelihood of ending up in any given situation (e.g., when all players have the same initial state distribution), the game of choosing contingency plans \textit{a priori} is totally symmetric. (\appref{app:veil_of_ignorance} gives an example.) For its analog in the philosophy of \citet{rawls1971theory} and \citet{harsanyi1975can}, we call this situation the \textit{veil of  ignorance}.

\subsection{What do Symmetric Games Look Like?} \label{sec:symmetry_illustration}

To illustrate types of symmetry in games, Figure~\ref{fig:symmetry_illustration} presents symmetric variants of a self-driving taxi grid-world game inspired by the motivating example of \secref{sec:laundry_washing_up}. The taxis can move to adjacent grid cells, and they are on a team to drive people around a town with a residential area and a business area.

An idealized symmetric environment is shown in Figure~\ref{fig:symmetric_mdp_a}. Here, the self-driving taxis are identical, and the environment is perfectly symmetric; the symmetry of the game is clear. This is the sort of symmetry that might be found in highly controlled environments such as factories.

Identical agents in an asymmetric environment are shown in Figure~\ref{fig:symmetric_mdp_b}. Because the self-driving taxis are identical and have the same initial condition, their action sequences can be swapped without changing the outcome of the game. Thus, \textit{the game is symmetric even though the environment is asymmetric}. While it is impossible for real-world agents to have the exact same physical location, it suffices for them to have the same \textit{distribution} over initial conditions. Furthermore, we expect that \textit{virtual} agents (such as customer service chatbots or nodes in a compute cluster) may have identical initial conditions.

Nonidentical agents in an asymmetric environment are shown in Figure~\ref{fig:symmetric_mdp_c}. If we assume that the type and / or the initial location of each self-driving taxi is equally random, then the game of choosing contingency plans behind the veil of ignorance (\secref{sec:veil_of_ignorance}) is totally symmetric. We expect this case of symmetry to be common when AI uses the same source code or the same learned parameters. In fact, weight sharing is a common practice in multi-agent RL \citep{foerster2016learning}.

\section{Local Symmetric Optima are (Global) Nash Equilibria} \label{sec:theoretical}

After the formal definitions of symmetry in the previous section, we are almost ready to formally state the first of our main results. The only remaining definition is that of a local symmetric optimum:

\begin{restatable}{definition}{deflocallyoptimalstrategy}
Call $s$ a \textit{locally optimal $\mathcal{P}$-invariant} strategy profile of a common-payoff game if: (i) $s$ is $\mathcal{P}$-invariant, and (ii) for some $\epsilon > 0$, no $\mathcal{P}$-invariant strategy $s'$ with $EU(s') > EU(s)$ can be formed by adding or subtracting at most $\epsilon$ to the probability of taking any given action $a_i \in A_i$. If, furthermore, condition (ii) holds \textit{for all} $\epsilon > 0$, we call $s$ a \textit{globally optimal} $\mathcal{P}$\textit{-invariant} strategy profile or simply an \textit{optimal} $\mathcal{P}$\textit{-invariant} strategy profile.
\end{restatable}

Now we can state our first main theorem, that local symmetric optima are (global) Nash equilibria:

\begin{restatable}{theorem}{symlocaloptisnash}
\label{thm:efg_equilibrium}
Let $\Gg$ be a common-payoff normal-form game, and let $\mathcal{P} \subseteq \Gamma(\Gg)$ be a subset of the game symmetries $\Gamma(\Gg)$. Any locally optimal $\mathcal{P}$-invariant strategy profile is a Nash equilibrium.
\end{restatable}

\begin{proof}
We provide a sketch here and full details in \appref{app:main_proof}.
Suppose, for the sake of contradiction, that an individual player $i$ could beneficially deviate to action $a_i$ (if a beneficial deviation exists, then there is one to a pure strategy). Then, consider instead a collective change to a symmetric strategy profile in which all the players in $i$'s orbit shift slightly more probability to $a_i$. By making the amount of probability shifted ever smaller, the  probability that this change affects exactly one agent's realized action (making it $a_i$ when it would not have been before) can be arbitrarily larger than the probability that it affects multiple agents' realized actions. Moreover, if this causes exactly one agent's realized action to change, this must be in expectation beneficial, since the original unilateral deviation was in expectation beneficial. Hence, the original strategy profile cannot have been locally optimal. 
\end{proof}

\subsection{Applications of the Theorem}

First, we provide an example of applying \Cref{thm:efg_equilibrium} to multi-agent RL.

\begin{restatable}{example}{exmultiagentrl}
\label{ex:marl_iterative_best_response}
Consider a cooperative multi-agent RL environment where all agents have the same initial state distribution. Suppose, as is typical practice \citep{foerster2016learning}, that we use a gradient method to train the parameters of a policy that all agents will share. Assume that the gradient method reaches a symmetric local optimum in mixed strategy space. If we wanted to improve upon this symmetric local optimum, we might lift the symmetry requirement and perform iterative best response, i.e., continue learning by updating the parameters of just one agent. However, by \Cref{thm:efg_equilibrium}, the symmetric local optimum is a Nash equilibrium. Thus, updating the parameters of a single agent cannot improve the common payoff; updating the parameters of at least two agents is necessary.
\end{restatable}

The preceding example assumes that a gradient method in multi-agent RL reaches a symmetric local optimum in mixed strategy space. In practice, agents may employ behavioral strategies, and it may not be possible to verify how close a symmetric strategy profile is to a local optimum.

In Appendix~\ref{app:general_symmetry_example}, we give another example that shows how \Cref{thm:efg_equilibrium} is more general than the case of total symmetry. The example illustrates the existence of rotational symmetry without total symmetry, and it illustrates how picking different $\mathcal{P} \subseteq \Gamma(\Gg)$ leads to different optimal $\mathcal{P}$-invariant strategies and thus different $\mathcal{P}$-invariant Nash equilibria by \Cref{thm:efg_equilibrium}.

\subsection{Robustness to Payoff and Strategy Perturbations} \label{sec:robustness}

\Cref{thm:efg_equilibrium} assumes that all players' payoffs are exactly the same, and it applies to strategy profiles that are exact local optima. If we relax these assumptions, the theorem still holds approximately. If all players' payoffs are equal $\pm \epsilon$, or if a strategy profile is $\epsilon$ distance away from a symmetric local optimum, then a robust version of \Cref{thm:efg_equilibrium} guarantees a $k\epsilon$-Nash equilibrium for some game-dependent constant $k$. See Appendix \ref{app:robustness} for a precise treatment of these robustness results.

While the results of this section concern Nash equilibria, we note that Nash equilibria, by definition, consider the possibility of only a \textit{single} agent deviating. In the next section, we investigate when multiple agents might have an incentive to \textit{simultaneously} deviate by studying the optimality of symmetric strategy profiles in possibly-asymmetric strategy space.

\subsection{Extending the Theorem to Multiple Teams}

So far, we have considered the cooperation of a single team. We can also study the interaction of many different teams. Suppose each team shares a common payoff while the interaction between the different teams is general sum. For example, we could extend the self-driving taxi game of \secref{sec:laundry_washing_up} to have multiple self-driving taxi companies. In prior work, the special case of zero-sum interaction between one team and a single adversary is called an \textit{adversarial team game} \citep{von1997team,celli2018computational,carminati2022public}.

\Cref{thm:efg_equilibrium} directly translates to this setting with multiple teams. To see why, consider a metagame with one player for each team. In the metagame, each of the metaplayers controls the strategy profile of their team. Suppose we are at a (local) Nash equilibrium in this metagame and that each team is playing a strategy profile that is $\mathcal{P}$-invariant in the original game. Now consider just one team trying to update to improve its payoff. If we leave the strategies of the other teams fixed, then this becomes a single-team, common-payoff game. So \Cref{thm:efg_equilibrium} applies, and no individual player can deviate to improve their payoff. By repeating this argument for every team, we see that no individual player on any team can deviate to improve their payoff. Therefore, the individual players of the original game are in a (global) Nash equilibrium.

\section{When are Local Optima in Symmetric Strategy Space also Local Optima in Possibly-asymmetric Strategy Space?} \label{sec:asymmetric_stability}

Our preceding theory applies to locally optimal $\mathcal{P}$-invariant, i.e., symmetric, strategy profiles. This leaves open the question of how well locally optimal symmetric strategy profiles perform when considered in the broader, possibly-asymmetric strategy space. When are locally optimal $\mathcal{P}$-invariant strategy profiles also locally optimal in possibly-asymmetric strategy space? This question is important in machine learning (ML) applications where users of symmetrically optimal ML systems might be motivated to make modifications to the systems, even for purposes of a common payoff.

To address this precisely, we formally define a \textit{local optimum in possibly-asymmetric strategy space}:
\begin{restatable}{definition}{defasymmetricoptimum}
\label{def:asymmetric_optimum}
A strategy profile $s = (s_1, s_2, \ldots, s_n)$ of a common-payoff normal-form game is \textit{locally optimal among possibly-asymmetric strategy profiles}, or, equivalently, a \textit{local optimum in possibly-asymmetric strategy space}, if for some $\epsilon > 0$, no strategy profile $s'$ with $EU(s') > EU(s)$ can be formed by changing $s$
in such a way that the probability of taking any given action $a_i \in A_i$ for any player $i$ changes by at most $\epsilon$.
\end{restatable}

\begin{restatable}{remark}{reminstability}
\label{rem:instability}
\Cref{def:asymmetric_optimum} relates to notions of \textit{stability} under dynamics, such as those with perturbations or stochasticity, that allow multiple players to make asymmetric deviations. In particular, if $s$ is not a local maximum in asymmetric strategy space, this means that there is some set of players $C$ and strategy $s'_C$ arbitrarily close to $s$, such that if players $C$ were to play $s'_C$ (by mistake or due to stochasticity), some Player $i\in N-C$ would develop a strict preference over the support of $s_i$. To illustrate this, we return to the motivating example of self-driving taxis.
\end{restatable}

\begin{restatable}{example}{exwashingupinstability}
\label{ex:washing_up_instability}
Consider again the game of \tabref{laundry-washing-symmetric-unstable}. As Figure \ref{fig:anticoordination_landscape} illustrates, the symmetric optimum is for both Auto and Mobile to randomize uniformly between H and W. While this is a Nash equilibrium, it is not a local optimum in possibly-asymmetric strategy space. If one player deviates from uniformly randomizing, the other player develops a strict preference for either $H$ or $W$.
\end{restatable}

To generalize the phenomenon of \Cref{ex:washing_up_instability}, we use the following \textit{degeneracy}\footnote{We note that ``degnerate'' is already an established term in the game-theoretic literature where it is often applied only to two-player games 
\citep[see, e.g,][Definition 3.2]{vonStengel2007}. While similar to the established notion of degeneracy, our definition of degeneracy is stronger, which makes our statements about non-degenerate games more general. (See \appref{app:nondegeneracy_combined} for details.)} condition:
\begin{restatable}{definition}{defnondegeneracy}
\label{def:nondegeneracy}
Let $s$ be a Nash equilibrium of a game $\Gg$: (i) If $s$ is deterministic, i.e., if every $s_i$ is a Dirac delta function on some $a_i$, then $s$ is \textit{degenerate} if at least two players $i$ are indifferent between $a_i$ and some other $a_i' \in A_i - \{a_i\}$. (ii) Otherwise, if $s$ is mixed, then $s$ is \textit{degenerate} if for all players $i$ and all $a_{-i}\subseteq \mathrm{supp} (s_{-i})$, the term $ EU_i(a_i,a_{-i})$ is constant across $a_i\in \mathrm{supp}(s_i)$.

We call a game $\Gg$ \textit{degenerate} if it has at least one degenerate Nash equilibrium.
\end{restatable}
Intuitively, our definition says that a deterministic Nash equilibrium is non-degenerate when it is strict or almost strict (excepting of at most one player who may be indifferent over available actions). A mixed Nash equilibrium, on the other hand, is non-degenerate when \textit{mixing matters}.

In non-degenerate games, our next theorem shows that a local symmetric optimum is a local optimum in possibly-asymmetric strategy space if and only if it is deterministic. Formally:
\begin{restatable}{theorem}{thmnondegeneracycombined}
\label{thm:nondegeneracy_combined}
Let $\Gg$ be a non-degenerate common-payoff normal-form game, and let $\mathcal{P} \subseteq \Gamma(\Gg)$ be a subset of the game symmetries $\Gamma(\Gg)$.
A locally optimal $\mathcal{P}$-invariant strategy profile is locally optimal among possibly-asymmetric strategy profiles if and only if it is deterministic.
\end{restatable}

\begin{wraptable}{r}{0.27\textwidth}
\vspace{-.1in}
\begin{tabular}{c|c|c|c|}
    \multicolumn{1}{c}{} & \multicolumn{1}{c}{a} & \multicolumn{1}{c}{b} & \multicolumn{1}{c}{c} \\\cline{2-4}
    a & 1 & 1 & 1 \\\cline{2-4}
    b & 1 & -10 & $1+\epsilon$ \\\cline{2-4}
    c & 1 & $1+\epsilon$ & -10\\\cline{2-4}
\end{tabular}
\vspace{-.1in}
\end{wraptable}
To see why the non-degeneracy condition is needed in \Cref{thm:nondegeneracy_combined}, we provide an example of a degenerate game:

\begin{example}
Consider the 3x3 symmetric game shown above.
Here, $(a,a)$ is the unique global optimum in symmetric strategy space. By \Cref{thm:efg_equilibrium}, it is therefore also a Nash equilibrium. However, it is a degenerate Nash equilibrium and not locally optimal in asymmetric strategic space. The payoff can be improved by, e.g., the row player shifting small probability to $b$, and the column player shifting small probability to $c$.
\end{example}

We have already seen an example of a non-degenerate deterministic equilibrium. The symmetric optimum from \secref{sec:complex_coordination_example}, even though it is not the \textit{global} asymmetric optimum, is nevertheless \textit{locally} optimal in possibly-asymmetric strategy space by \Cref{thm:nondegeneracy_combined}.

\section{Learning Symmetric Strategies in GAMUT} \label{sec:gamut_replicator}

\Cref{thm:nondegeneracy_combined} shows that, in non-degenerate games, a locally optimal symmetric strategy profile is stable in the sense of \Cref{rem:instability} if and only if it is pure. For those concerned about stability, this raises the question: how often are optimal strategies pure, and how often are they mixed?

To answer this question, we present an empirical analysis of learning symmetric strategy profiles in the GAMUT suite of game generators \citep{nudelman2004run}. We are interested both in how centralized optimization algorithms (such as gradient methods) search for symmetric strategies and in how decentralized populations of agents evolve symmetric strategies. To study the former, we run Sequential Least SQuares Programming (SLSQP) \citep{kraft1988software,2020SciPy-NMeth}, a local search method for constrained optimization. To study the latter, we simulate the replicator dynamics \citep{fudenberg1998theory}, an update rule from evolutionary game theory with connections to reinforcement learning \citep{borgers1997learning, tuyls2003extended, tuyls2003selection}. (See \appref{app:replicator_dynamics} for details.)

\subsection{Experimental Setup} \label{sec:gamut_experimental_setup}

We ran experiments in all three classes of symmetric GAMUT games: RandomGame, CoordinationGame, and CollaborationGame. (While other classes of GAMUT games, such as the prisoner's dilemma, exist, they cannot be turned into a symmetric, common-payoff game without losing their essential structure.) Intuitively, a RandomGame draws all payoffs uniformly at random, whereas in a CoordinationGame and a CollaborationGame, the highest payoffs are always for outcomes where all players choose the same action. (See \appref{app:gamut_definitions} for details.) Because CoordinationGame and CollaborationGame have such similar game structures, our experimental results in the two games are nearly identical. To avoid redundancy, we only include experimental results for CoordinationGame.

For each game class, we sweep the parameters of the game from 2 to 5 players and 2 to 5 actions, i.e., with $(|N|, |A_i|) \in \{2, 3, 4, 5\} \times \{2, 3, 4, 5\}$. We sample 100 games at each parameter setting and then attempt to calculate the global symmetric optimum using (i) 10 runs of SLSQP and (ii) 10 runs of the replicator dynamic (each with a different initialization drawn uniformly at random over the simplex), resulting in 10 + 10 = 20 solution attempts per game. Because we do not have ground truth for the globally optimal solution of the game (which is NP-hard to compute), we instead use the best of our 20 solution attempts, which we call the ``best solution.''

\subsection{How Often are Symmetric Optima Local Optima among Possibly-asymmetric Strategies?} \label{sec:gamut_instability}

Here, we try to get a sense for how often symmetric optima are stable in the sense that they are also local optima in possibly-asymmetric strategy space (see \Cref{rem:instability}). In Appendix \tabref{tab:coordinationgame_max_mixed}, we show in what fraction of games the best solution of our 20 optimization attempts is mixed; by \Cref{thm:nondegeneracy_combined} and \Cref{prp:gamut_measure_zero_degen} from the Appendix, this is the fraction of games whose symmetric optima are not local optima in possibly-asymmetric strategy space. In CoordinationGames, the symmetric optimum is always (by construction) for all players to choose the same action, leading to stability. By contrast, we see that 36\% to 60\% of RandomGames are \textit{unstable}. We conclude that if real-world games do not have the special structure of CoordinationGames, then instability may be common.

\subsection{How Often do SLSQP and the Replicator Dynamic Find an Optimal Solution?} \label{sec:gamut_optimality}

As sequential least squares programming and the replicator dynamic are not guaranteed to converge to a global optimum, we test empirically how often each run converges to the best solution of our 20 optimization runs. In Appendix \tabref{tab:global_one_run_optimality} / \tabref{tab:replicator_one_run_optimality}, we show what fraction of the time any single SLSQP / replicator dynamics run finds the best solution, and in Appendix \tabref{tab:global_many_run_optimality} / \tabref{tab:replicator_many_run_optimality}, we show what fraction of the time at least 1 of 10 SLSQP / replicator dynamics runs finds the best solution. First, we note that the tables for SLSQP and the replicator dynamics are quite similar, differing by no more than a few percentage points in all cases. So the replicator dynamics, which are used as a model for how populations evolve strategies, can also be used as an effective optimization algorithm. Second, we see that individual runs of each algorithm are up to 93\% likely to find the best solution in small RandomGames, but they are less likely (as little as 24\% likely) to find the best solution in larger RandomGames and in CoordinationGames. The best of 10 runs, however, finds the best solution $\geq 87\%$ of the time; so random algorithm restarts benefit symmetric strategy optimization.

\section{Conclusion}

There are a variety of reasons we expect to see symmetric games in machine learning systems. The first is mass hardware production, which will proliferate identical robots such as self-driving cars, that require ad-hoc cooperation \citep{stone2010ad}. The second is interaction over the internet, where websites treat all users equally. The third is anonymous protocols, such as voting, which depend on symmetry. As Figure~\ref{fig:symmetry_illustration} shows, symmetric \textit{games} can still arise even when \textit{agents} and the \textit{environment} are asymmetric.

Similarly, there are a variety of reasons we expect to see symmetric strategies in practice. The first is software copies: we expect many artificial agents will run the same source code. The second is optimization - enforcing symmetric strategies exponentially reduces the joint-strategy space. The third is parameter sharing between different neural networks, which can be critical to success in multi-agent RL \citep{foerster2016learning} and may occur as a result of pretraining on large datasets \citep{dasari2020robonet}. The fourth is communication: symmetry (and symmetry breaking) is a key component of zero-shot coordination with other agents and humans \citep{hu2020other,Treutlein2021}. The fifth is that a single-player game with imperfect recall can be interpreted as a multi-agent game in symmetric strategies \citep{Aumann1997}.

When cooperative AI is deployed in the world with symmetric strategy profiles, it raises questions about the properties of such profiles. Would individual agents (or the users they serve) want to deviate from these profiles? Are they robust to small changes in the game or in the executed strategies? Could there be better asymmetric strategy profiles nearby?

Our results yield a mix of good and bad news. Theorems \ref{thm:efg_equilibrium} and \ref{thm:tv_bound} are good news for stability, showing that even local optima in symmetric strategy space are (global) Nash equilibria in a robust sense.  So, with respect to \textit{unilateral} deviations among team members, symmetric optima are relatively stable. On the other hand, this may be bad news for optimization because unilateral deviation cannot improve on a local symmetric optimum (\Cref{ex:marl_iterative_best_response}). Furthermore, \Cref{thm:nondegeneracy_combined} is perhaps bad news, showing that a broad class of symmetric local optima are unstable when considering \textit{joint} deviations in asymmetric strategy space (\Cref{rem:instability}). Empirically, our results with learning algorithms in GAMUT suggest that these unstable solutions may not be uncommon in practice (Section \ref{sec:gamut_instability}).

Future work could build on our analysis in a few different ways. First, we focus on \textit{mixed} strategy space. However, future work may wish to deal with \textit{behavioral} strategy space. While mixed and behavioral strategies are equivalent in games of perfect recall \citep{kuhn1953extensive,aumann1961mixed}, they are not equivalent for games of imperfect recall \citep{piccione1997interpretation}. Second, we focus on players who can play an arbitrary mixed strategy over discrete actions. Future work could consider continuous action space and players who act according to a learned probability distribution. We expect learned probability distributions to pose an additional challenge because, for some probability distributions, our proof of \Cref{thm:efg_equilibrium} in \appref{app:main_proof} will not directly transfer. Our proof requires that agents can always adjust a strategy by moving arbitrarily small probability onto a single action. However, this is not possible with many distributions such as Gaussian distributions. Finally, our experimental results focus on the normal-form representation of games in GAMUT \citep{nudelman2004run}. It would be interesting to see what experimental properties symmetric optima have in sequential decision making benchmarks.

\section*{Acknowledgements}

We thank Stephen Casper, Lawrence Chan, Michael Dennis, Frances Ding, Daniel Filan, Rachel Freedman, Jakob Foerster, Adam Gleave, Rohin Shah, Sam Toyer, Alex Turner, and the anonymous reviewers for helpful feedback on this work.

We are grateful for the support we received for this work. It includes NSF Award IIS-1814056, funding from the DOE CSGF under grant number DE-SC0020347, and funding from the Cooperative AI Foundation, the Center for Emerging Risk Research, the Berkeley Existential Risk Initiative, and the Open Philanthropy Foundation. We also appreciate the Leverhulme Trust’s support for the Centre for the Future of Intelligence.

\bibliography{main}
\bibliographystyle{icml2022}

\newpage
\appendix
\onecolumn
\section{Veil of Ignorance Example} \label{app:veil_of_ignorance}

Two robots arrive at a resource that can be used by only one of them.  They can choose as their action either Cautious or Aggressive. If both choose C, one of them gets the resource at random. If exactly one chooses A, that one gets the resource. If both choose A, the resource is destroyed and neither gets it (utility 0).

Each robot privately knows whether it has High or Low need for the resource (each type occurs independently with probability $1/2$). A robot that has High need values the resource at 6; one that has Low need values it at 4. Robots are on the same team and care about the sum of utilities.

From behind the veil of ignorance, the optimal symmetric strategy (contingency plan) is: when having type L, always play C; when having type H, play A with probability $p=1/6$ (and C otherwise). Note, as guaranteed by \Cref{thm:efg_equilibrium}, that this is a Nash equilibrium. To verify this, observe that from the perspective of a robot with type H, the expected team utility for playing A (when the other follows the given strategy with $p$) is $(1/2) \cdot 6 + (1/2)(1-p) \cdot 6 = 6 - 3p$, and for playing C it is $(1/2)((4+6)/2) + (1/2)\cdot 6 = 5.5$, and if $p=1/6$ these are equal. In contrast, from the perspective of a robot with type L, the expected team utility for playing A (when the other follows the given strategy) is $(1/2) \cdot 4 + (1/2)(1-p) \cdot 4 = 4 - 3p$, and for playing C it is $(1/2)\cdot 4 + (1/2)(p \cdot 6 + (1-p) \cdot (4+6)/2)= 4.5 + p/2$, so C is strictly preferred.

Overall, this optimal symmetric strategy results in an expected team utility of $(1/4) \cdot 4 + (1/2) \cdot ((5/6) \cdot 4 + (1/6) \cdot 6) + (1/4) \cdot ((5/6)(5/6) \cdot 4 + 2(1/6)(5/6) \cdot 6 + (1/36) \cdot 0) = 77/18 \approx 4.28$. Compare this with an asymmetric strategy where robot 1 plays A when it has type H but otherwise C is always played by both robots, which results in a team utility of $(4+5+6+6)/4 = 21/4 = 5.25$. (If types were not private knowledge, $22/4=11/2=5.5$ would be possible.)

In this example, we see how players can coordinate using symmetric strategies from behind the veil of ignorance. Although it is possible to achieve a higher payoff using asymmetric strategies, the optimal symmetric strategy is nonetheless a Nash equilibrium by \Cref{thm:efg_equilibrium}.

\section{Proofs of \secref{sec:theoretical} Results} \label{app:main_proof}

\symlocaloptisnash*
\begin{proof}
We proceed by contradiction. Suppose $s = (s_1, s_2, \ldots, s_n)$ is locally optimal among $\mathcal{P}$-invariant strategy profiles but is not a Nash equilibrium. We will construct an $s'$ arbitrarily close to $s$ with $EU(s') > EU(s)$.

Without loss of generality, suppose $s_1$ is not a best response to $s_{-1}$ but that the pure strategy of always playing $a_1$ is a best response to $s_{-1}$. For an arbitrary probability $p>0$, consider the modified strategy $s_1'$ that plays action $a_1$ with probability $p$ and follows $s_1$ with probability $1 - p$. Now, construct $s' = (s_1', s_2', \ldots, s_n')$ as follows:
\begin{align*}
    s_i' = \begin{cases} s_i' = s_1' & \text{if $i \in \mathcal{P}(1)$} \\ s_i' = s_i & \text{otherwise.} \end{cases}
\end{align*}
In words, $s'$ modifies $s$ by having the members of player $1$'s orbit mix in a probability $p$ of playing $a_1$. We claim for all sufficiently small $p$ that $EU(s') > EU(s)$.

To establish this claim, we break up the expected utility of $s'$ according to cases of how many players in $1$'s orbit play the action $a_1$ because of mixing in $a_1$ with probability $p$. In particular, we observe
\begin{equation*}
    \begin{split}
    EU(s') &= B(m {=} 0, p) EU(s)\\
    &\ \ \ + B(m {=} 1, p) EU((s_1', s_2, \ldots, s_n))\\
    &\ \ \ + B(m {>} 1, p) EU(\ldots),
    \end{split}
\end{equation*}
where $B(m, p)$ is the probability of $m$ successes for a binomial random variable on $m$ independent events that each have success probability $p$ and where $EU(\ldots)$ is arbitrary. Note that the crucial step in writing this expression is grouping the terms with the coefficient $B(m {=} 1, p)$. We can do this because for any player $j \in \mathcal{P}(1)$, there exists a symmetry $\rho \in \Gamma(\Gg)$ with $\rho(j) = 1$.

Now, to achieve $EU(s') > EU(s)$, we require
\begin{equation*}
\begin{split}
EU(s) &< \frac{B(m = 1, p)}{B(m > 0, p)}EU((s_1', s_2, \ldots, s_n))\\
&\ \ \ + \frac{B(m > 1, p)}{B(m > 0, p)}EU(...).
\end{split}
\end{equation*}
We know $EU((s_1', s_2, ..., s_n)) > EU(s)$, but we must deal with the case when $EU(...)$ is arbitrarily negative. Because $\lim_{p \to 0} B(m > 1, p) / B(m = 1, p) = 0$, by making $p$ sufficiently small, $B(m = 1, p) / B(m > 0, p)$ can be made greater than $B(m > 1, p) / B(m > 0, p)$ by an arbitrarily large ratio. The result follows.
\end{proof}

\section{Example of General Symmetry in \Cref{thm:efg_equilibrium}}
\label{app:general_symmetry_example}

\begin{restatable}{example}{exradiostations}
\label{ex:radio_stations}
There are four %
groups of partygoers positioned in a square. %
We number these 1,2,3,4 clockwise, such that, e.g., 1 neighbors 4 and 2.
There is also a robot butler at each vertex of the square. The partygoers can fetch refreshments from the robot butler at their vertex of the square and from the robot butler at adjacent vertices of the square, but it is too far of a walk for them to fetch refreshments from the robot at the opposite vertex.

The game has each robot butler choose what refreshment to hold. For simplicity, suppose each robot butler can hold food or drink. The common payoff of the game is the sum of the utilities of the four groups of partygoers. For each group, if the group cannot fetch drink, the payoff for that group is 0. If the group can only fetch drink, the payoff is 1, and if the group can fetch food and drink, the group's payoff is 2.

The symmetries of the game $\Gamma(\Gg)$ include the set of permutations generated by rotating the robot butlers once clockwise. In standard notation for permutations, $\{ (1,2,3,4),(2,3,4,1),(3,4,1,2), (4,1,2,3) \} \subset \Gamma(\Gg)$.

First, consider applying the theorem to $\mathcal{P} = \Gamma(\Gg)$. In this case, the constraint of $\mathcal{P}$-invariance requires all the robot butlers play the same strategy because all of them are in the same orbit. As we show in the proof below, the optimal $\mathcal{P}$-invariant strategy is then for each robot to hold food with probability $\sqrt{2}-1$. \Cref{thm:efg_equilibrium} tells us that this optimal $\mathcal{P}$-invariant strategy profile is a Nash equilibrium. The proof below also shows how to verify this without the use of \Cref{thm:efg_equilibrium}.

Second, consider applying the theorem to the case where $\mathcal{P}$ consists only of the rotation twice clockwise, i.e., the permutation which maps each robot onto the robot on the opposite vertex of the square. In standard notation for permutations, $\mathcal{P}=\{(3,4,1,2) \}$.
Now, the constraint of $\mathcal{P}$-invariance requires robot butlers at opposite vertices of the square to play the same strategy. However, neighboring robots can hold different refreshments. The optimal $\mathcal{P}$-invariant strategy is for one pair of opposite-vertex robots, e.g., 1 and 3, to hold food and for the other pair of robots, 2 and 4, to hold drink. While it turns out to be immediate that this optimal $\mathcal{P}$-invariant strategy is a Nash equilibrium because it achieves the globally optimal outcome, we could have applied \Cref{thm:efg_equilibrium} to know that this optimal $\mathcal{P}$-invariant strategy profile is a Nash equilibrium \textit{even without knowing what the optimal $\mathcal{P}$-invariant strategy was}.
\end{restatable}
\begin{proof}
We here calculate the optimal $\Gamma(\mathcal{G})$-invariant strategy profile for \Cref{ex:radio_stations}. Let $p$ be the probability of holding drink. By symmetry of the game and linearity of expectation, the expected utility given $p$ is simply four times the expected utility of any one group of partygoers. The utility of one group of partygoers is $0$ with probability $(1-p)^3$, is $1$ with probability $p^3$ and is $2$ with the remaining probability. Hence, the expected utility of a single group of partygoers is
\begin{equation*}
   p^3 + (1- (1-p)^3 - p^3)\cdot 2 = 2 - 2(1-p)^3-p^3.
\end{equation*}
The maximum of this term (and thus the maximum of the overall utility of all neighborhoods) can be found by any computer algebra system %
to be $p=2-\sqrt{2}$, which gives an expected utility of $4(\sqrt{2}-1)\approx 1.66$.

To double-check, we can also calculate the symmetric Nash equilibrium of this game. It's easy to see that that Nash equilibrium must be mixed and therefore must make each robot butler indifferent about what to hold. 
So let $p$ again be the probability with which each robot butler holds drink.
The expected utility of holding drink relative to holding nothing for any of the three relevant neighborhoods is $2(1-p)^2$.
(Holding drink lifts the utility of a group of partygoers from 0 to 2 if they can not already fetch drink. Otherwise, it doesn't help to hold drink.)
The expected utility of holding food relative to broadcasting nothing is simply $p^2$. We can find the symmetric Nash equilibrium by setting
\begin{equation*}
    2(1-p)^2=p^2,
\end{equation*}
which %
gives
us the same solution for $p$ as before.
\end{proof}

\section{Robustness of \Cref{thm:efg_equilibrium} to Payoff and Strategy Perturbations} \label{app:robustness}

The first type of robustness we consider is robustness to perturbations in the game's payoff function. Formally, we define an \textit{$\epsilon$-perturbation} of a game as follows:

\begin{restatable}{definition}{defepsperturbation}
\label{def:eps_perturbation}
Let $\Gg$ be a normal-form game with utility function $\mu$. For some $\epsilon > 0$, we call $\Gg'$ an \textit{$\epsilon$-perturbation} of $\Gg$ if $\Gg'$ has utility function $\mu'$ satisfying $\max_{i \in N, a \in A} \lvert u_i'(a) - u_i(a) \rvert \leq \epsilon$.
\end{restatable}

There are a variety of reasons why $\epsilon$-perturbations might arise in practice. Our game model may contain errors such as the game not being perfectly symmetric; the players' preferences might fluctuate over time; or we might have used function approximation to learn the game's payoffs. With \Cref{prp:payoff_perturbed_eps_nash}, we note a generic observation about Nash equilibria showing that our main result, \Cref{thm:efg_equilibrium}, is robust in the sense of degrading \textit{linearly} in the payoff perturbation's size:

\begin{restatable}{proposition}{prppayoffperturbedepsnash}
\label{prp:payoff_perturbed_eps_nash}
Let $\Gg$ be a common-payoff normal-form game, and let $s^*$ be a locally-optimal $\mathcal{P}$-invariant strategy profile for some subset of game symmetries $\mathcal{P} \subseteq \Gamma(\Gg)$. Suppose $G'$ is an $\epsilon$-perturbation of $\Gg$. Then $s^*$ is a $2\epsilon$-Nash equilibrium in $\Gg'$.
\end{restatable}
\begin{proof}
By \Cref{thm:efg_equilibrium}, $s^*$ is a Nash equilibrium in $\Gg$. After perturbing $\Gg$ by $\epsilon$ to form $\Gg'$, payoffs have increased / decreased at most $\pm \epsilon$, so the difference between any two actions' expected payoffs has changed by at most $2\epsilon$.
\end{proof}

The second type of robustness we consider is robustness to symmetric solutions that are only approximate. For example, we might try to find a symmetric local optimum through an approximate optimization method, or the evolutionary dynamics among players' strategies might lead them to approximate local symmetric optima. Again, a generic result about Nash equilibria shows that the guarantee of \Cref{thm:efg_equilibrium} degrades linearly in this case:
\begin{restatable}{theorem}{thmtvbound}
\label{thm:tv_bound}
Let $\Gg$ be a common-payoff normal-form game, and let $s^*$ be a locally-optimal $\mathcal{P}$-invariant strategy profile for some subset of game symmetries $\mathcal{P} \subseteq \Gamma(\Gg)$. Suppose $s$ is a strategy profile with total variation distance $TV(s, s^*) \leq \delta$. Then $s$ is an $\epsilon$-Nash equilibrium with $\epsilon = 4 \delta \max_{i \in N, a \in A} \lvert u_i(a) \rvert$.
\end{restatable}
\begin{proof}
Consider the perspective of an arbitrary player $i$. The difference in expected utility of playing any action $a_i$ between the opponent strategy profiles $s_{-i}$ and $s_{-i}^*$ is given by:
\begin{align*}
&\ \ \ \left \lvert EU_i(a_i, s_{-i}) - EU_i(a_i, s_{-i}^*) \right \rvert \\
&= \left \lvert \sum_{a_{-i} \in A_{-i}} s_{-i} (a_{-i}) u_i(a_i, a_{-i}) \right. \\
& \ \ \ - \left. \sum_{a_{-i} \in A_{-i}} s_{-i}^* (a_{-i}) u_i(a_i, a_{-i}) \right \rvert \\
&\leq \sum_{a_{-i} \in A_{-i}} \lvert u_i(a_i, a_{-i}) \rvert \left \lvert s_{-i}(a_{-i}) - s_{-i}^* (a_{-i}) \right \rvert \\
&\leq 2 TV(s, s^*) \max_{i \in N, a \in A} \lvert u_i(a) \rvert \\
&\leq 2 \delta \max_{i \in N, a \in A} \lvert u_i(a) \rvert.
\end{align*}

In particular, let $a_i$ be an action in the support of $s_i^*$, and let $a_i'$ be any other action. Then, using the above, we have:
\begin{align*}
    &\ \ \ EU_i(a_i', s_{-i}) - EU_i(a_i, s_{-i}) \\
    &= EU_i(a_i', s_{-i}) - EU_i(a_i', s_{-i}^*) + EU_i(a_i', s_{-i}^*) \\
    &\ \ \  - EU_i(a_i, s_{-i}^*) + EU_i(a_i, s_{-i}^*) - EU_i(a_i, s_{-i}) \\
    &\leq EU_i(a_i', s_{-i}) - EU_i(a_i', s_{-i}^*)\\
    &\ \ \ + EU_i(a_i, s_{-i}^*) - EU_i(a_i, s_{-i}) \\
    &\leq \left \lvert EU_i(a_i', s_{-i}) - EU_i(a_i', s_{-i}^*) \right \rvert \\
    &\ \ \  + \left \lvert EU_i(a_i, s_{-i}) - EU_i(a_i, s_{-i}^*) \right \rvert \\
    &\leq 4 \delta \max_{i \in N, a \in A} \lvert u_i(a) \rvert,
\end{align*}
where $EU_i(a_i', s_{-i}^*) - EU_i(a_i, s_{-i}^*) \leq 0$ because $s_i^*$ is a Nash equilibrium by \Cref{thm:efg_equilibrium}.
\end{proof}

By \Cref{thm:tv_bound}, we have a robustness guarantee in terms of the total variation distance between an approximate local symmetric optimum and a true local symmetric optimum. Without much difficulty, we can also convert this into a robustness guarantee in terms of the Kullback-Leibler divergence:

\begin{restatable}{corollary}{corklbound}
Let $\Gg$ be a common-payoff normal-form game, and let $s^*$ be a locally-optimal $\mathcal{P}$-invariant strategy profile for some subset of game symmetries $\mathcal{P} \subseteq \Gamma(\Gg)$. Suppose $s$ is a strategy profile with Kullback-Leibler divergence satisfying $D_{KL}(s \vert\vert s^*) \leq \nu$ or $D_{KL}(s^* \vert\vert s) \leq \nu$. Then $s$ is an $\epsilon$-Nash equilibrium with $\epsilon = 2 \sqrt{2 \nu} \max_{i \in N, a \in A} \lvert u_i(a) \rvert$.
\end{restatable}
\begin{proof}
By Pinsker's inequality \citep{tsybakov2009introduction}, we have
\begin{align*}
    TV(s, s^*) &\leq \sqrt{\frac{1}{2} D_{KL}(s \vert\vert s^*)}.
\end{align*}
As $TV(s, s^*) = TV(s^*, s)$ and with a similar application of Pinsker's inequality, we have by assumption that $TV(s, s^*) \leq \sqrt{\nu / 2}$. Applying \Cref{thm:tv_bound} with $\delta = \sqrt{\nu / 2}$ yields the result.
\end{proof}

\section{Proof of \secref{sec:asymmetric_stability} Results}
\label{app:nondegeneracy_combined}

First, we clarify how our notion of non-degeneracy compares to the existing literature. If a two-player game $\Gg$ is non-degenerate in the usual sense from the literature, it is non-degenerate in the sense of \Cref{def:nondegeneracy}. Moreover, if $\Gg$ is common-payoff, then for each player $i$, we can define a two-player game played by $i$ and another single player who controls the strategies of $N - \{i\}$. If for all $i$ these two-player games are non-degenerate in the established sense, then $\Gg$ is non-degenerate in the sense of \Cref{def:nondegeneracy}.

Now, we proceed with the proof of \secref{sec:asymmetric_stability} results:

\thmnondegeneracycombined*
\begin{proof}
Let $s$ be a locally optimal $\mathcal{P}$-invariant strategy profile. By \Cref{thm:efg_equilibrium}, $s$ is a Nash equilibrium. Because $\Gg$ is non-degenerate, so is $s$. We prove the claim by proving that (1) if $s$ is deterministic, it is locally optimal in asymmetric strategy space; and (2) if $s$ is mixed then it is not locally optimal in asymmetric strategy space.\\

(1) The deterministic case:
Let $s$ be deterministic. Now consider a potentially asymmetric strategy profile $s'$. We must show as $s'$ becomes sufficiently close to $s$ that $EU(s')\leq EU(s)$.

Let $\epsilon_1,\epsilon_2,...,\epsilon_n$ and $\hat s_1,...,\hat s_n$ be such that for $i\in N$, $s_i'$ can be interpreted as following $s_i$ with probability $1-\epsilon_i$ and following $\hat s_i$ with probability $\epsilon_i$, where $s_i \notin \mathrm{supp}(\hat s_i)$. Then (similar to the proof of \Cref{thm:efg_equilibrium}), we can write
\begin{align*}
     &\quad\ \ EU(s')\\
     &= \left( \prod_{i\in N} (1-\epsilon_i) \right) EU(s)\\
     &+ \sum_{j \in N} \epsilon_j \left( \prod_{i\in N-\{ j \}} 1-\epsilon_i \right) \cdot EU(\hat s_j, s_{-j})\\
     &+ \sum_{j,l\in N:j\neq l} \epsilon_j \epsilon_i \left( \prod_{i\in N-\{ j,l \}} 1-\epsilon_i \right) \cdot EU(\hat s_j, \hat s_l, s_{-j-l})\\
     &+ ...
\end{align*}
The second line is the expected value if everyone plays $s$, the third line is the sum over the possibilities of one player $j$ deviating to $\hat s_j$, and so forth. We now make two observations. First, because $s$ is a Nash equilibrium, the expected utilities $EU(\hat s_j, s_{-j})$ in the third line are all at most as big as $EU(s)$. Now consider any later term corresponding to the deviation of some set $C$, containing at least two players $i,j$. Note that it may be $EU(\hat s_C, s_{-C})>EU(s)$. However, this term is multiplied by $\epsilon_i\epsilon_j$. Thus, as the $\epsilon$ go to $0$, the significance of this term in the average vanishes in comparison to that of both the terms corresponding to the deviation of just $i$ and just $j$, which are multiplied only by $\epsilon_i$ and $\epsilon_j$, respectively. By non-degeneracy, it is $EU(\hat s_i, s_{-i})<EU(s)$ or $EU(\hat s_j, s_{-j})<EU(s)$. Thus, if the $\epsilon_i$ are small enough, the overall sum is less than $EU(s)$.

(2) The mixed case: Let $s$ be mixed.
We proceed by constructing a strategy profile $s'$ that is arbitrarily close to $s$ with $EU(s') > EU(s)$.

Let $m$ be the largest integer where for all subsets of players $C \subseteq N$ with $\lvert C \rvert \leq m$, the expected payoff is constant across all joint deviations to $a_i \in \mathrm{supp}(s_i)$ for all $i \in C$, i.e., where $EU(a_C, s_{-C}) = EU(s)$ for all $a_C \in \mathrm{supp}(s_C)$. As $s$ is a non-degenerate Nash equilibrium, $1 \leq m < n$.

By definition of $m$, there exists a subset of players $C \subset N$ with $\lvert C \rvert = m$ and choice of actions $a_C \in \mathrm{supp}(s_C)$ where $EU(a_j, a_C, s_{-j-C})$ is not constant across the available actions $a_j \in A_j$ for some player $j \in N - C$. Denote player $j$'s best response to the joint deviation $a_C$ as $a_j^* \in \argmax_{a_j} EU(a_j, a_C, s_{-j-C})$, and note $EU(a_j, a_C, s_{-j-C}) > EU(a_C,s_{-C}) = EU(s)$.

To construct $s'$, modify $s$ by letting player $j$ mix according to $s_j$ with probability $(1 - \epsilon)$ and play action $a_j$ with probability $\epsilon$. Similarly, let each player $i \in C$ mix according to $s_i$ with probability $(1 - \epsilon)$ and play their action $a_i$ specified by $a_C$ with probability $\epsilon$. Because we allow $\epsilon > 0$ to be arbitrarily small, all we have left to show is $EU(s') > EU(s)$.

Observe as before that we can break $EU(s')$ up into cases based on the number of players who deviate according to the modified probability $\epsilon$:
\begin{align*}
    &\quad\ \ EU(s')\\
    &= \left( \prod_{k\in C \cup \{j\}} (1-\epsilon) \right) EU(s)\\
    &+ \sum_{l \in C \cup \{j\}} \epsilon \left( \prod_{k \in C \cup \{ j \} : k \neq l}  1-\epsilon \right) EU(a_l, s_{-l})\\
    &+ ... \\
    &+ \left( \prod_{k \in C \cup \{j\}} \epsilon \right) EU(a_j, a_C, s_{-j-C}).
\end{align*}
By construction, every value in the expected value calculation $EU(s')$ is equal to $EU(s)$ except for the last value $EU(a_j, a_C, s_{-j-C})$, which is greater than $EU(s)$. We conclude $EU(s') > EU(s)$.
\end{proof}

\section{GAMUT Details and Additional Experiments} \label{app:gamut_replicator}

\subsection{GAMUT Games} \label{app:gamut_definitions}
In \secref{sec:gamut_experimental_setup}, we analyzed all three classes of symmetric GAMUT games: RandomGame, CoordinationGame, and CollaborationGame. Below, we give a formal definiton of these game classes:

\begin{restatable}{definition}{defrandomgame}
A \textit{RandomGame} with $|N|$ players and $|A|$ actions assumes $A_i = A_j$ for all $i, j$ and draws a payoff from $\mathit{Unif}(-100, 100)$ for each unordered action profile $a \in A$.
\end{restatable}

\begin{restatable}{definition}{defcoordinationgame}
A \textit{CoordinationGame} with $|N|$ players and $|A|$ actions assumes $A_i = A_j$ for all $i, j$. For each unordered action profile $a \in A$ with $a_i = a_j$ for all $i, j$, it draws a payoff from $\mathit{Unif}(0, 100)$; for all other unordered action profiles, it draws a payoff from $\mathit{Unif}(-100, 0)$.
\end{restatable}

\begin{restatable}{definition}{defcollaborationgame}
A \textit{CollaborationGame} with $|N|$ players and $|A|$ actions assumes $A_i = A_j$ for all $i, j$. For each unordered action profile $a \in A$ with $a_i = a_j$ for all $i, j$, the payoff is 100; for all other unordered action profiles, it draws a payoff from $\mathit{Unif}(-100, 99)$.
\end{restatable}

Note that these games define payoffs for each \textit{unordered} action profile because the games are totally symmetric (\Cref{def:total_symmetry}).
\tabref{tab:gamut_examples} gives illustrative examples.

\begin{table}
    \begin{center}
    \setlength{\extrarowheight}{2pt}
    \begin{tabular}{cc|c|c|}
      & \multicolumn{1}{c}{} & \multicolumn{2}{c}{Player $2$}\\
      & \multicolumn{1}{c}{} & \multicolumn{1}{c}{$\alpha$}  & \multicolumn{1}{c}{$\beta$} \\\cline{3-4}
      \multirow{2}*{Player $1$}  & $\alpha$ & $u_{\alpha \alpha}$ & $u_{\alpha \beta}$ \\\cline{3-4}
      & $\beta$ & $u_{\alpha \beta}$ & $u_{\beta \beta}$ \\\cline{3-4}
    \end{tabular}
    \caption{A payoff matrix with $|N| = 2$ and $A_1 = A_2 = \{\alpha, \beta\}$ to illustrate GAMUT games. In a RandomGame, $u_{\alpha \alpha}$, $u_{\alpha \beta}$, and $u_{\beta \beta}$ are i.i.d. draws from $\mathit{Unif}(-100, 100)$. In a CoordinationGame, $u_{\alpha \alpha}$ and $u_{\beta \beta}$ are i.i.d. draws from $\mathit{Unif}(0, 100)$ while $u_{\alpha \beta}$ is a draw from $\mathit{Unif}(-100, 0)$. In a CollaborationGame, $u_{\alpha \alpha} = u_{\beta \beta} = 100$, and $u_{\alpha \beta}$ is a draw from $\mathit{Unif}(-100, 99)$.}
    \label{tab:gamut_examples}
    \end{center}
\end{table}

\subsection{Proof of Non-degeneracy in GAMUT}

\begin{restatable}{proposition}{prpgamutmeasurezerodegen}
\label{prp:gamut_measure_zero_degen}
Drawing a degenerate game is a measure-zero event in RandomGames, CoordinationGames, and CollaborationGames, i.e., these games are almost surely non-degenerate.
\end{restatable}
\begin{proof}
By \Cref{def:nondegeneracy}, in order for a game to be degenerate, there must exist a player $i$, a set of actions for the other players $a_{-i}$, and a pair of actions $a_i \neq a_i'$ with $EU(a_i, a_{-i}) = EU(a_i', a_{-i})$. In RandomGames, CoordinationGames, and CollaborationGames, $EU(a_i, a_{-i}) = \mu(a_i, a_{-i})$ and $EU(a_i', a_{-i}) = \mu(a_i', a_{-i})$ are continuous random variables that are independent of each other. (Or, in the case of a CollaborationGame, $\mu(a_i, a_{-i})$ may be a fixed value outside of the support of $\mu(a_i', a_{-i})$.) So $EU(a_i, a_{-i}) = EU(a_i', a_{-i})$ is a measure-zero event.
\end{proof}

\subsection{Replicator Dynamics} \label{app:replicator_dynamics}

Consider a game where all players share the same action set, i.e., with $A_i = A_j$ for all $i, j$, and consider a totally symmetric strategy profile $s = (s_1, s_1, \ldots, s_1)$. In the replicator dynamic, each action can be viewed as a species, and $s_1$ defines the distribution of each individual species (action) in the overall population (of actions). At each iteration of the replicator dynamic, the prevalence of an individual species (action) grows in proportion to its relative fitness in the overall population (of actions). In particular, the replicator dynamic evolves $s_1(a)$ over time $t$ for each $a \in A_1$ as follows:
\begin{equation*}
    \frac{d}{dt} s_1(a) = s_1(a) \left[ EU(a, s_{-1}) - EU(s) \right].
\end{equation*}
To simulate the replicator dynamic with Euler's method, we need to choose a stepsize and a total number of iterations. Experimentally, we found the fastest convergence with a stepsize of 1, and we found that 100 iterations sufficed for convergence; see Figure \ref{fig:replicator_dynamics_convergence}. For good measure, we ran 10,000 iterations of the replicator dynamic in all of our experiments.

\begin{figure}[t]
\centering
\includegraphics[width=0.5\columnwidth]{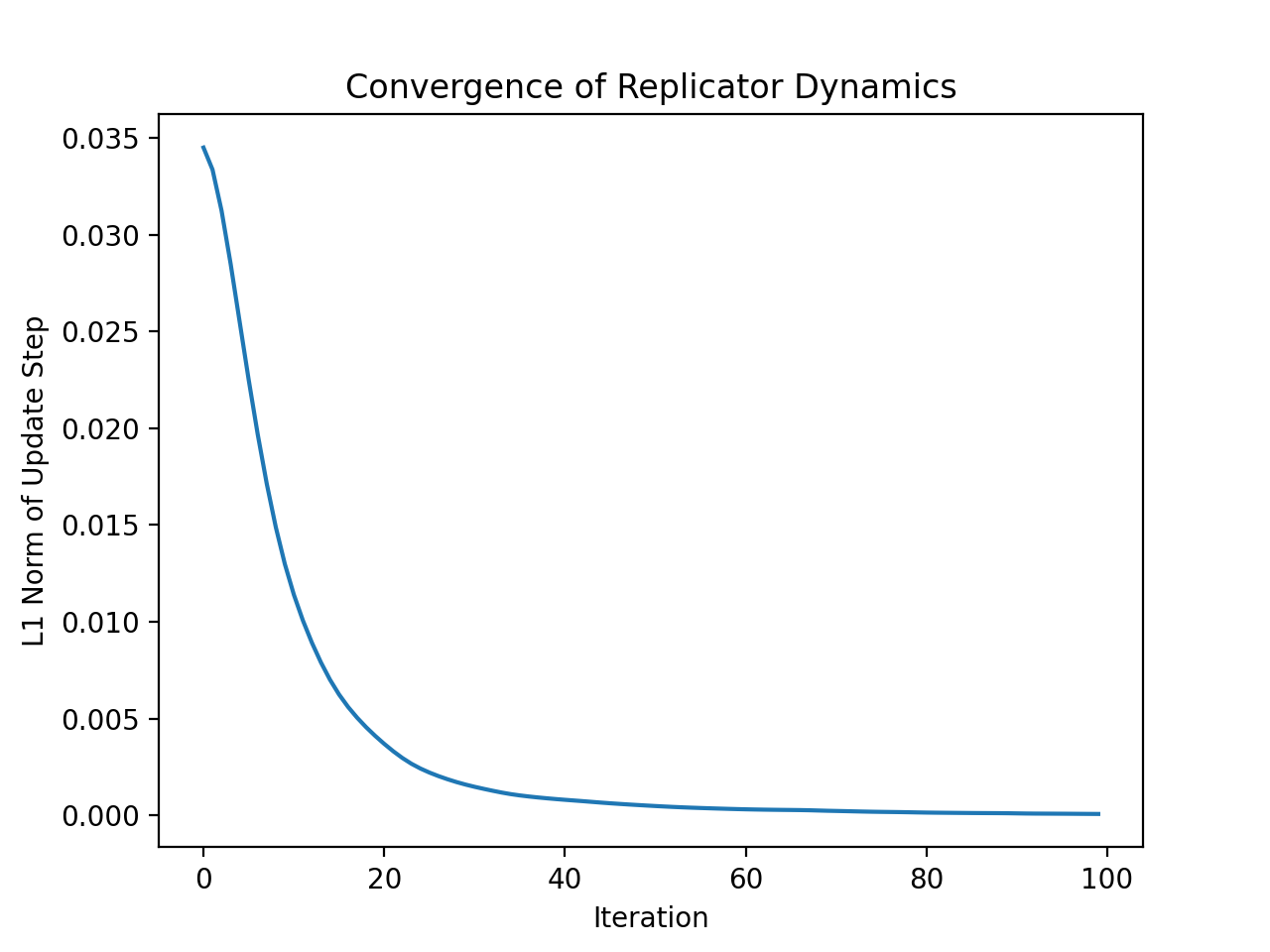}
\caption{The magnitude of the replicator dynamics update step averaged over 10,000 RandomGames\protect\footnotemark \ with 2 players and 2 actions. Although this plot indicates that the replicator dynamics converge by 100 iterations, we ran 10,000 iterations for good measure in all of our experiments.}
\label{fig:replicator_dynamics_convergence}
\end{figure}

\footnotetext{In this simulation only we rescaled the RandomGames so that each payoff is a draw from $\mathit{Unif}(0, 1)$.}

We are interested in the replicator dynamic for two reasons. First, it is a model for how agents in the real world may collectively arrive at a symmetric solution to a game (e.g., through evolutionary pressure). Second, it is a learning algorithm that performs local search in the space of symmetric strategies. In our experiments of Appendix \ref{app:gamut_algorithm_optimality}, we find that using the replicator dynamic as an optimization algorithm is competitive with Sequential Least SQuares Programming (SLSQP), a local search method from the constrained optimization literature \citep{kraft1988software,2020SciPy-NMeth}.

\subsection{What Fraction of Symmetric Optima are Local Optima in Possibly-asymmetric Strategy Space?} \label{app:gamut_instability}

As discussed in \secref{sec:gamut_instability}, we would like to get a sense for how often symmetric optima are stable in the sense that they are also local optima in possibly-asymmetric strategy space (see \Cref{rem:instability}). \tabref{tab:max_mixed} shows in what fraction of games the best solution we found is \textit{unstable}.

\begin{table*}
    \centering
    \subfloat[RandomGame\label{tab:randomgame_max_mixed}]{
    \begin{tabular}{lrrrr}
\toprule
A &    2 &    3 &    4 &    5 \\
N &      &      &      &      \\
\midrule
2 & 0.36 & 0.44 & 0.44 & 0.50 \\
3 & 0.38 & 0.49 & 0.59 & 0.60 \\
4 & 0.42 & 0.45 & 0.46 & 0.46 \\
5 & 0.45 & 0.48 & 0.49 & 0.47 \\
\bottomrule
\end{tabular}

}
\subfloat[CoordinationGame\label{tab:coordinationgame_max_mixed}]{
    \begin{tabular}{lrrrr}
\toprule
A &  2 &  3 &  4 &  5 \\
N &    &    &    &    \\
\midrule
2 &  0 &  0 &  0 &  0 \\
3 &  0 &  0 &  0 &  0 \\
4 &  0 &  0 &  0 &  0 \\
5 &  0 &  0 &  0 &  0 \\
\bottomrule
\end{tabular}

}
        \caption{The fraction of games whose symmetric optima are mixed. By \Cref{thm:nondegeneracy_combined}, these symmetric equilibria are the ones \textit{unstable} in the sense of \Cref{rem:instability}. Numbers in the table were empirically determined from 100 randomly sampled games per GAMUT class.}
    \label{tab:max_mixed}
\end{table*}

\subsection{How Often do SLSQP and the Replicator Dynamic Find an Optimal Solution?} \label{app:gamut_algorithm_optimality}

As discussed in \secref{sec:gamut_optimality}, \tabref{tab:global_one_run_optimality} and \tabref{tab:global_many_run_optimality} show how often SLSQP finds an optimal solution, while \tabref{tab:replicator_one_run_optimality} and \tabref{tab:replicator_many_run_optimality} show how often the replicator dynamic finds an optimal solution.

\begin{table*}
    \centering
    \subfloat[RandomGame\label{tab:randomgame_global_one_run_optimality}]{
    \begin{tabular}{lrrrr}
\toprule
A &    2 &    3 &    4 &    5 \\
N &      &      &      &      \\
\midrule
2 & 0.92 & 0.81 & 0.70 & 0.64 \\
3 & 0.80 & 0.69 & 0.57 & 0.48 \\
4 & 0.75 & 0.57 & 0.40 & 0.35 \\
5 & 0.70 & 0.45 & 0.36 & 0.31 \\
\bottomrule
\end{tabular}

}
\subfloat[CoordinationGame\label{tab:coordinationgame_global_one_run_optimality}]{
    \begin{tabular}{lrrrr}
\toprule
A &    2 &    3 &    4 &    5 \\
N &      &      &      &      \\
\midrule
2 & 0.59 & 0.50 & 0.40 & 0.33 \\
3 & 0.53 & 0.38 & 0.28 & 0.29 \\
4 & 0.53 & 0.37 & 0.29 & 0.26 \\
5 & 0.53 & 0.36 & 0.33 & 0.25 \\
\bottomrule
\end{tabular}

}
        \caption{The fraction of single SLSQP runs that achieve the best solution found in our 20 total optimization attempts. Numbers in the table were empirically determined from 100 randomly sampled games per GAMUT class.}
    \label{tab:global_one_run_optimality}
\end{table*}

\begin{table*}
    \centering
    \subfloat[RandomGame\label{tab:randomgame_global_many_run_optimality}]{
    \begin{tabular}{lrrrr}
\toprule
A &    2 &    3 &    4 &    5 \\
N &      &      &      &      \\
\midrule
2 & 1.00 & 0.99 & 0.99 & 0.98 \\
3 & 1.00 & 0.99 & 1.00 & 0.96 \\
4 & 1.00 & 0.96 & 0.94 & 0.88 \\
5 & 0.98 & 0.90 & 0.88 & 0.91 \\
\bottomrule
\end{tabular}

}
\subfloat[CoordinationGame\label{tab:coordinationgame_global_many_run_optimality}]{
    \begin{tabular}{lrrrr}
\toprule
A &    2 &    3 &    4 &    5 \\
N &      &      &      &      \\
\midrule
2 & 0.99 & 1.00 & 0.98 & 0.97 \\
3 & 1.00 & 0.99 & 0.93 & 0.95 \\
4 & 1.00 & 0.97 & 0.97 & 0.93 \\
5 & 0.99 & 1.00 & 0.95 & 0.92 \\
\bottomrule
\end{tabular}

}
        \caption{The fraction of games in which at least 1 of 10 SLSQP runs achieves the best solution found in our 20 total optimization attempts. Numbers in the table were empirically determined from 100 randomly sampled games per GAMUT class.}
    \label{tab:global_many_run_optimality}
\end{table*}

\begin{table*}
    \centering
    \subfloat[RandomGame\label{tab:randomgame_replicator_one_run_optimality}]{
    \begin{tabular}{lrrrr}
\toprule
A &    2 &    3 &    4 &    5 \\
N &      &      &      &      \\
\midrule
2 & 0.93 & 0.81 & 0.68 & 0.65 \\
3 & 0.81 & 0.70 & 0.58 & 0.46 \\
4 & 0.76 & 0.58 & 0.36 & 0.34 \\
5 & 0.69 & 0.43 & 0.36 & 0.30 \\
\bottomrule
\end{tabular}

}
\subfloat[CoordinationGame\label{tab:coordinationgame_replicator_one_run_optimality}]{
    \begin{tabular}{lrrrr}
\toprule
A &    2 &    3 &    4 &    5 \\
N &      &      &      &      \\
\midrule
2 & 0.58 & 0.45 & 0.40 & 0.33 \\
3 & 0.57 & 0.35 & 0.29 & 0.27 \\
4 & 0.53 & 0.37 & 0.28 & 0.25 \\
5 & 0.51 & 0.33 & 0.33 & 0.24 \\
\bottomrule
\end{tabular}

}
        \caption{The fraction of single replicator dynamics runs that achieve the best solution found in our 20 total optimization attempts. Numbers in the table were empirically determined from 100 randomly sampled games per GAMUT class.}
    \label{tab:replicator_one_run_optimality}
\end{table*}

\begin{table*}
    \centering
    \subfloat[RandomGame\label{tab:randomgame_replicator_many_run_optimality}]{
    \begin{tabular}{lrrrr}
\toprule
A &    2 &    3 &    4 &    5 \\
N &      &      &      &      \\
\midrule
2 & 1.00 & 1.00 & 1.00 & 1.00 \\
3 & 0.99 & 1.00 & 0.95 & 0.96 \\
4 & 1.00 & 0.98 & 0.91 & 0.91 \\
5 & 0.98 & 0.97 & 0.92 & 0.87 \\
\bottomrule
\end{tabular}

}
\subfloat[CoordinationGame\label{tab:coordinationgame_replicator_many_run_optimality}]{
    \begin{tabular}{lrrrr}
\toprule
A &    2 &    3 &    4 &    5 \\
N &      &      &      &      \\
\midrule
2 & 1.00 & 1.00 & 0.99 & 0.94 \\
3 & 1.00 & 0.97 & 0.93 & 0.96 \\
4 & 0.99 & 1.00 & 0.93 & 0.92 \\
5 & 1.00 & 0.98 & 0.96 & 0.90 \\
\bottomrule
\end{tabular}

}
        \caption{The fraction games in which at least 1 of 10 replicator dynamics runs achieves the best solution found in our 20 total optimization attempts. Numbers in the table were empirically determined from 100 randomly sampled games per GAMUT class.}
    \label{tab:replicator_many_run_optimality}
\end{table*}

\subsection{How Costly is Payoff Perturbation under the Simultaneous Best Response Dynamic?} \label{app:gamut_epsilon_bribes}

When a game is not stable in the sense of \Cref{rem:instability}, we would like to understand how costly the worst-case $\epsilon$-perturbation of the game can be. (See \Cref{def:eps_perturbation} for the definition of an $\epsilon$-perturbation of a game.) In particular, we study the case when individuals simultaneously update their strategies in possibly-asymmetric ways by defining the following \textit{simultaneous best response dynamic}:
\begin{restatable}{definition}{defsimultaneousbestresponsedynamic}
The \textit{simultaneous best response dynamic at $s$} updates from strategy profile $s = (s_1, s_2, \ldots, s_n)$ to strategy profile $s' = (s_1', s_2', \ldots, s_n')$ with every $s_i'$ a best response to $s_{-i}$.
\end{restatable}
For each of the RandomGames in \secref{sec:gamut_instability} whose symmetric optimum $s$ is not a local optimum in possibly-asymmetric strategy space, we compute the worst-case $\epsilon$ payoff perturbation for infinitesimal $\epsilon$. Then, we update each player's strategy according to the simultaneous best response dynamic at $s$. This necessarily leads to a decrease in the original common payoff because the players take simultaneous updates on an objective that, after payoff perturbation, is no longer common. \tabref{tab:gamut_max_decrease} reports the average percentage decrease in expected utility, which ranges from 55\% to 89\%. Our results indicate that simultaneous best responses after payoff perturbation in RandomGames can be quite costly.

\begin{table*}
    \centering
    \subfloat[RandomGame\label{tab:randomgame_max_decrease}]{
    \begin{tabular}{lrrrr}
\toprule
A &     2 &     3 &     4 &     5 \\
N &       &       &       &       \\
\midrule
2 & 58.9\% & 55.9\% & 61.8\% & 64.6\% \\
3 & 73.7\% & 70.9\% & 73.4\% & 73.7\% \\
4 & 74.1\% & 77.4\% & 78.4\% & 82.5\% \\
5 & 77.4\% & 84.9\% & 89.9\% & 87.5\% \\
\bottomrule
\end{tabular}

}

        \caption{The average decrease in expected utility that worst-case infinitesimal asymmetric payoff perturbations cause to unstable symmetric optima. To get these numbers, we first perturb payoffs in the 100 RandomGames from \secref{sec:gamut_instability} whose symmetric optima $s$ are not local optima in possibly-asymmetric strategy space. Then, in each perturbed game, we compute a simultaneous best-response update to $s$ and record its decrease in expected utility.}
    \label{tab:gamut_max_decrease}
\end{table*}

\section{Code and Computational Resources} \label{app:code_and_compute}

All of our code is available at \url{https://github.com/scottemmons/coordination} under the MIT License. With a reduced number of random seeds, we guess that it would be possible to reproduce the experiments in this paper on a modern laptop. To test a large number of random seeds, we ran our experiments for a few days on an Amazon Web Services \texttt{c5.24xlarge} instance.

Our code uses the following Python libraries:
\begin{itemize}
    \item Matplotlib \citep{Hunter:2007}, released under ``a nonexclusive, royalty-free, world-wide license,''
    \item NumPy \citep{harris2020array}, released under the BSD 3-Clause ``New'' or ``Revised'' License,
    \item pandas \citep{jeff_reback_2021_4681666,mckinney-proc-scipy-2010}, released under the BSD 3-Clause ``New'' or ``Revised'' License,
    \item SciPy \citep{2020SciPy-NMeth}, released under the BSD 3-Clause ``New'' or ``Revised'' License, and
    \item SymPy \citep{10.7717/peerj-cs.103}, released under the New BSD License.
\end{itemize}

\end{document}